\documentclass[copyright,creativecommons]{eptcs}
\providecommand{\event}{Infinity 2010} 

\usepackage{breakurl}

\usepackage{tikz}
\usetikzlibrary{shapes}
\usepackage[T1]{fontenc}
\usepackage[utf8x]{inputenc}

\usepackage{amsmath}
\usepackage{amssymb}
\usepackage{amsthm}

\usepackage{negarrow}
\usepackage{fancybox}

\usepackage{proof}

\usepackage{xcolor}

\usepackage{array}

\usepackage{comment}

\definecolor{lightblue}{rgb}{0.55,0.55,1.0}

\newcommand{\defs}{\stackrel{{\scriptscriptstyle \fun{def}}}{=}}

\renewcommand{\iota}{}

\newcommand{\piout}[2]{\overline{#1}{\langle{#2}\rangle}}
\newcommand{\piin}[2]{#1(#2)}
\newcommand{\pim}[2]{[#1=#2]}
\newcommand{\sep}{\shortmid}
\newcommand{\equi}{{\leftrightarrow}}
\newcommand{\compat}[1]{{\stackrel{#1}{\leftrightarrow}}}

\newcommand{\fun}[1]{\mathsf{#1}}

\newcommand{\powerset}[1]{\mathbb{P}({#1})}

\newcommand{\redex}[1]{\fbox{\ensuremath{#1}}}

\newcommand{\bisim}{\sim}

\newcommand{\llbracket}{\ensuremath{[\![}}
\newcommand{\rrbracket}{\ensuremath{]\!]}}

\newtheorem{theorem}{Theorem}
\newtheorem{lemma}{Lemma}
\newtheorem{definition}{Definition}
\newtheorem{proposition}{Proposition}

\newcommand{\suchthat}{\mid}

\newcommand{\vect}[1]{\overrightarrow{#1}}
\newif\ifcomment
\commenttrue 
\definecolor{gris}{gray}{0.3}

\newcommand{\drop}[1]{}

\title{A Decidable Characterization\\ of a Graphical Pi-calculus with Iterators
\\} 
\author{Fr\'ed\'eric Peschanski
\institute{UPMC -- LIP6}
\and Hanna Klaudel\institute{Universit\'e Evry--Ibisc}
\and Raymond Devillers\institute{Universit\'e Libre de Bruxelles}}

\def\titlerunning{A Graphical Pi-calculus with Iterators}
\def\authorrunning{F. Peschanski, H. Klaudel \& R. Devillers}

\begin{document}
\maketitle

\begin{abstract}
  This paper presents the Pi-graphs, a visual paradigm for the
  modelling and verification of mobile systems. The language is a
  graphical variant of the Pi-calculus with iterators to express
  non-terminating behaviors. The operational semantics of Pi-graphs
  use ground notions of labelled transition and bisimulation, which
  means standard verification techniques can be applied. We show that
  bisimilarity is decidable for the proposed semantics, a result
  obtained thanks to an original notion of causal clock as well as the
  automatic garbage collection of unused names.
\end{abstract}

\section{Introduction}



The $\pi$-graphs is a visual paradigm loosely inspired by the
Petri nets. It is a graphical variant of the
$\pi$-calculus~\cite{pi:book} with similar constructs and semantics.
The formalism is designed as both a \emph{modelling} language and a
\emph{verification} framework. 

The design of a graphical modelling
language has subjective motivations: intuitiveness, aesthetics, etc.
One design choice we retain from mainstream visual languages (UML,
Petri nets, etc.) is \emph{staticness}: the preservation of the
diagrammatic structure along transitions. Most graphical
interpretations of the $\pi$-calculus involve dynamic diagrams:
nodes and edges are created/deleted along
transitions~\cite{DBLP:conf/esop/Milner94,DBLP:journals/njc/Parrow95,gadducci:pigraphenc:mscs:2006}.
In contrast, the structure of the $\pi$-graphs does not evolve over
time. The idea is to ``move'' names around a static graph, using an
inductive variant of \emph{graph
  relabelling}~\cite{graph:relabelling:handbook}. For non-terminating
behaviors, we use \emph{iterators}~\cite{DBLP:conf/icalp/BusiGZ04},
 a suitable static substitute for control-finite recursion.

 Beyond modelling, our second axle of research is verification with
 more objective goals. One difficulty is that the usual semantic
 variants of the $\pi$-calculus (early, late, open) rely on non-ground
 transition systems and/or bisimulation relations, which leads to
 specific and rather non-trivial verification techniques
 e.g.~\cite{DBLP:journals/iandc/PistoreS01,DBLP:conf/cav/VictorM94,DBLP:conf/cav/FerrariGMPR98}.
 The $\pi$-graphs, on the contrary, use ground notions of labelled
 transition and bisimulation, which means \emph{standard} verification
 techniques can be applied.  Of course, there is no magic, the
 ``missing'' information is recorded somewhere. First, each
 $\pi$-graph state is attached to a \emph{clock}. As explained in
 \cite{DBLP:conf/sofsem/PeschanskiB09}, the clock is used for the
 generation of names that are guaranteed fresh by construction. It is
 also used to characterize a form of \emph{read-write
   causality}~\cite{DBLP:conf/icalp/DeganoP95}. Moreover, the match
 and synchronization constructs are interpreted as the dynamic
 construction of a \emph{partition} deciding equality for names. 
 
  There are, however, two sources of \emph{infinity} in the proposed
 model. First, the logical clocks (used in \cite{DBLP:conf/sofsem/PeschanskiB09})
  can grow infinitely. Moreover,
   the generated fresh names are never reclaimed. This means that infinite state spaces
    can be constructed even for very simple iterative behaviors.
  To avoid the construction of
 infinite state spaces, we first introduce 
 an original (and non-trivial) model of causal clocks, which provide a more structured
  characterization of read-write causality. 
 As a second ``counter-measure'' against infinity, we develop an automatic garbage collection scheme for unused names in  graphs.
 As a
 major result, we show that bisimilarity is decidable for the proposed
 semantics. 
 
The outline of the paper is as follows. In Section~\ref{sec:language} we introduce the diagram language
and the corresponding process algebra. In Section~\ref{sec:semantics} the operational semantics is
proposed. The finiteness results are developed in Section~\ref{sec:results}. Related work is
discussed in Section~\ref{sec:related}.

\section{The diagram language and process algebra}
\label{sec:language}

The $\pi$-graphs is a visual language inspired by
(elementary) 
Petri nets. The control flow is characterized by interconnected places
with token marks. A data-part models the names and channels used by
the processes to interact. This is realized by placeholders called
\emph{boxes} that can be instantiated by names.  Places and boxes
cannot be arranged arbitrarily, and the $\pi$-graphs must conform to
the syntax described in Table~\ref{tab:syntax} (see page~\pageref{tab:syntax}). 
The basic syntactic elements are roughly the ones of the $\pi$-calculus
(cf. \cite{pi:book}): input, output, silent action, non-deterministic
choice and parallel compositions. A notable difference is that most of
the constructs (even match, parallel and sum) are considered in prefix
position. Moreover, the process expressions must be suffixed by an
explicit \emph{termination} $0$.

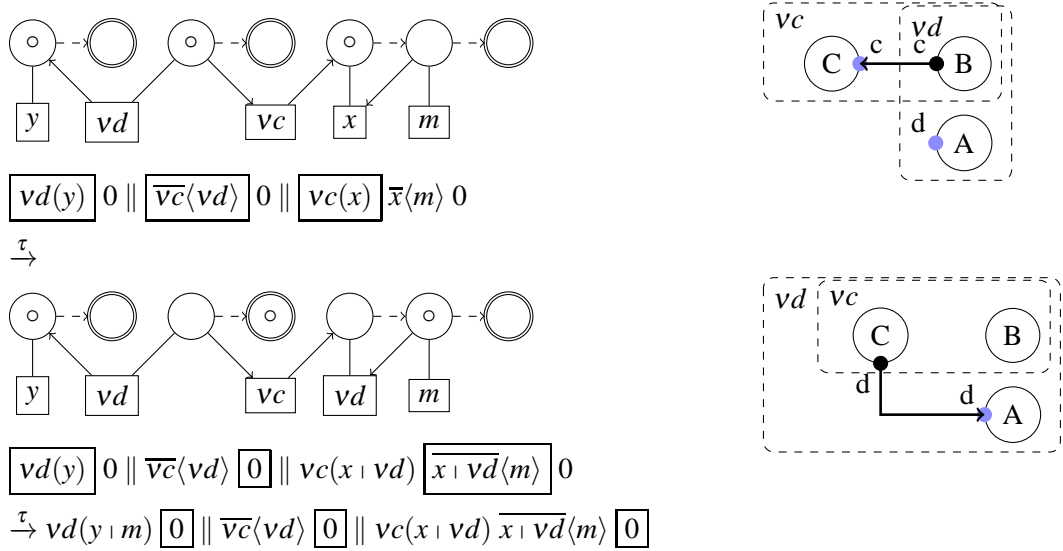
\begin{figure}[t]
\begin{center}
\begin{tabular}{m{0.6\textwidth}m{0.4\textwidth}}
   \begin{minipage}{0.6\textwidth}
      \begin{tikzpicture}[node distance=30pt]
      \node[circle,draw] (B) {$\circ$};
        \node[circle,draw,right of=B,node distance=60pt] (A) {$\circ$};
        \node[circle,draw,left of=B,node distance=60pt] (C) {$\circ$};
        \node[circle,draw,right of=B,double] (B2) {\phantom{$\circ$}};
        \node[rectangle,draw,below of=B2] (nc) {$\nu c$};
        \node[rectangle,draw, below of=A] (ix) {$\iota x$};
       \node[circle,draw,right of=A] (A2) {\phantom{$\circ$}};
    \node[circle,draw,right of=A2,double] (A3) {\phantom{$\circ$}};
    \node[rectangle,draw, below of=A2] (m) {$m$};
        \node[circle,draw,right of=C,double] (C2) {\phantom{$\circ$}};
        \node[rectangle,draw,below of=C2] (nd) {$\nu d$};
        \node[rectangle,draw, below of=C] (iy) {$\iota y$};
        \path[->] (B) edge (nc);
        \path[->] (nc) edge (A);
        \draw (C) to (iy);
        \draw (nd) to (B);
        \draw (A) to (ix);
        \draw (A2) to (m);
        \path[->] (nd) edge (C);
        \path[->] (A2) edge (ix);
       \draw[dashed,->] (A) to (A2);
        \draw[dashed,->] (A2) to (A3);
        \draw[dashed,->] (B) to (B2);
        \draw[dashed,->] (C) to (C2);
    \end{tikzpicture} \vspace{10pt}

      $\redex{\piin{\nu d}{\iota y}}~0 \parallel \redex{\piout{\nu c}{\nu d}}~0 
 \parallel \redex{\piin{\nu c}{\iota x}}~\piout{\iota x}{m}~0$
\end{minipage}
     &

\begin{tikzpicture}[node distance=50pt]
\node[circle,draw] (a) {C};
\fill[color=lightblue] (a.east) circle(.1cm);
\node[anchor=south west] (ac) at (a.east) {c};
\node[circle,draw,right of=a] (b) {B};
\fill[color=black] (b.west) circle(.1cm);
\node[anchor=south east] (bd) at (b.west) {c};
\node[node distance=20pt, below right of=b] (bb) {};
\node[circle,draw,below of=b, node distance=30pt] (c) {A};
\node[node distance=20pt, below right of=c] (cc) {};
\node[rectangle,anchor=north west, node distance=23pt, above left of=a] (pc) {${\nu}c$};
\draw[draw,dashed,rounded corners] (pc.north west) rectangle (bb);
\node[rectangle,anchor=north west, node distance=20pt, above left of=b] (pd)  {${\nu}d$};%
\draw[draw,dashed,rounded corners] (pd.north west) rectangle (cc.east);
\fill[color=lightblue] (c.west) circle(.1cm);
\node[anchor=south east] (cd) at (c.west) {d};
\draw[->,line width=1pt] (b) to (a);
\end{tikzpicture}

\\
$\xrightarrow{\tau}$
\\
      \begin{tikzpicture}[node distance=30pt]
      \node[circle,draw] (B) {\phantom{$\circ$}};
        \node[circle,draw,right of=B,node distance=60pt] (A) {\phantom{$\circ$}};
        \node[circle,draw,left of=B,node distance=60pt] (C) {$\circ$};
        \node[circle,draw,right of=B,double] (B2) {$\circ$};
        \node[rectangle,draw,below of=B2] (nc) {$\nu c$};
        \node[rectangle,draw, below of=A] (ix) {$\nu d$};
       \node[circle,draw,right of=A] (A2) {$\circ$};
    \node[circle,draw,right of=A2,double] (A3) {\phantom{$\circ$}};
    \node[rectangle,draw, below of=A2] (m) {$m$};
        \node[circle,draw,right of=C,double] (C2) {\phantom{$\circ$}};
        \node[rectangle,draw,below of=C2] (nd) {$\nu d$};
        \node[rectangle,draw, below of=C] (iy) {$\iota y$};
        \path[->] (B) edge (nc);
        \path[->] (nc) edge (A);
        \draw (C) to (iy);
        \draw (nd) to (B);
        \draw (A) to (ix);
        \draw (A2) to (m);
        \path[->] (nd) edge (C);
        \path[->] (A2) edge (ix);
       \draw[dashed,->] (A) to (A2);
        \draw[dashed,->] (A2) to (A3);
        \draw[dashed,->] (B) to (B2);
        \draw[dashed,->] (C) to (C2);
    \end{tikzpicture} \vspace{10pt}

$\redex{\piin{\nu d}{\iota y}}~0 \parallel \piout{\nu c}{\nu d}~\redex{0} 
 \parallel \piin{\nu c}{\iota x\sep \nu d}~\redex{\piout{\iota x\sep \nu d}{m}}~0$
& 
\begin{tikzpicture}[node distance=50pt]
\node[circle,draw] (a) {C};
\node[circle,draw,right of=a] (b) {B};
\fill[color=black] (a.south) circle(.1cm);
\node[anchor=north east] (ad) at (a.south) {d};
\node[node distance=20pt, below right of=b] (bb) {};
\node[circle,draw,below of=b, node distance=30pt] (c) {A};
\node[node distance=20pt, below right of=c] (cc) {};
\node[below of=a, node distance=30pt] (ac) {};
\node[rectangle,anchor=north west, node distance=20pt, above left of=a] (pc) {${\nu}c$};
\draw[draw,dashed,rounded corners] (pc.north west) rectangle (bb);
\node[rectangle,anchor=north west, node distance=20pt, left of=pc] (pd)  {${\nu}d$};%
\draw[draw,dashed,rounded corners] (pd.north west) rectangle (cc.east);
\fill[color=lightblue] (c.west) circle(.1cm);
\node[anchor=south east] (cd) at (c.west) {d};
\draw[->,line width=1pt] (a) -- (ac.center) -- (c);
\end{tikzpicture}

\\
 
\multicolumn{2}{l}{$\xrightarrow{\tau} \piin{\nu d}{\iota y \sep m}~\redex{0} \parallel \piout{\nu c}{\nu d}~\redex{0} 
 \parallel \piin{\nu c}{\iota x\sep \nu d}~\piout{\iota x\sep \nu d}{m}~\redex{0}$}

\end{tabular}
\end{center}
\caption{\label{fig:ex:mobility} Example with mobility (abridged)}
\end{figure}

As an illustration, consider the example of
Figure~\ref{fig:ex:mobility}. This is an archetype of the kind of
mobility involved in the $\pi$-calculus. The (extract of) $\pi$-graph
on the left describes three processes - $A$ (left), $B$ (center) and
$C$ (right) - evolving concurrently. The current state of each process is characterised by
a token mark in the corresponding place. In the term representation given below the
graph, each prefix in redex position corresponds to such a place with a
token mark. We depict this by surrounding the prefix with a frame. To
establish the link with the $\pi$-calculus, we added on the right a flowgraph
representation of the system (cf. \cite{pi:book}). The processes $B$ and $C$ share a private channel $\nu c$ and in the
first step, $B$ communicates with $C$ using this channel. The
transmitted data is another private channel $\nu d$, initially only
known by $A$ and $B$. We are thus in a situation of \emph{channel
  passing}. The process $C$ binds the received name (here $\nu d$) to
the box $\iota x$. In the term representation, the instantiation is
made explicit with the notation $\iota x \sep \nu d$.  The left name
is the identifier for the box in the graph, which is a \emph{static}
information, and the right name describes its \emph{dynamic}
instantiation. As a convenience, the default instantiation $n \sep n$
is simply denoted $n$.  The corresponding flowgraph shows the
\emph{scope extrusion} of the channel $\nu d$ so that it encompasses
$C$. In the last step there is a synchronization between $C$ and $A$
along $\nu d$ with the communication of the datum $m$ (only the term is depicted). 

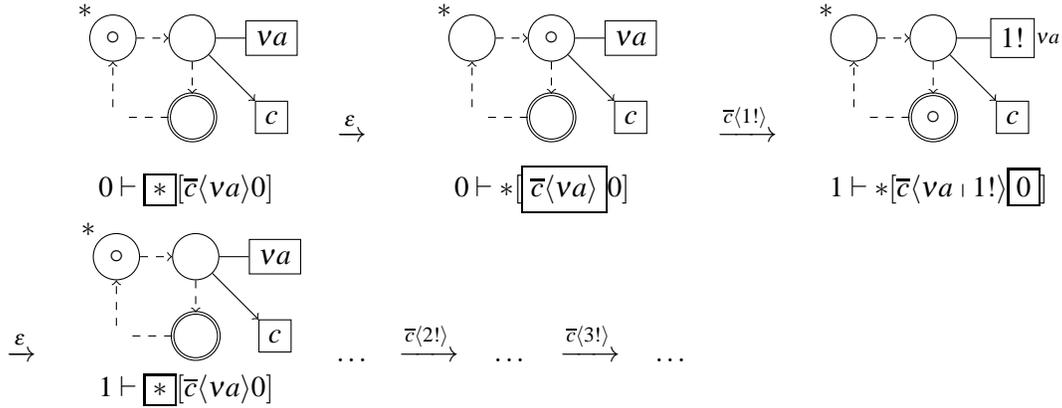
\begin{figure}[htb]
\begin{center}
\begin{tabular}{cccccccc}
    &
     \begin{tikzpicture}[node distance=30pt]
        \node[circle,draw] (iter) {$\circ$};
        \draw node[above left=4pt] at (iter) {$*$};
        \node[circle,draw,right of=iter] (A) {\phantom{$\circ$}};
        \node[circle,draw,below of=A,double] (B) {\phantom{$\circ$}};
        \node[rectangle,draw,right of=A] (na) {$\nu a$};
        \node[rectangle,draw, below of=na] (c) {$c$};
        \node[below of=iter] (loop) {};
        \path[->] (A) edge (c);
        \draw (A) to (na);
        \draw[dashed,->] (iter) to (A);
        \draw[dashed,->] (A) to (B);
        \draw[dashed] (B) to (loop);
        \draw[dashed,->] (loop) to (iter);
    \end{tikzpicture} \vspace{5pt}

    &
     $\xrightarrow{\varepsilon}$
    &
     \begin{tikzpicture}[node distance=30pt]
        \node[circle,draw] (iter) {\phantom{$\circ$}};
        \draw node[above left=4pt] at (iter) {$*$};
        \node[circle,draw,right of=iter] (A) {$\circ$};
        \node[circle,draw,below of=A,double] (B) {\phantom{$\circ$}};
        \node[rectangle,draw,right of=A] (na) {$\nu a$};
        \node[rectangle,draw, below of=na] (c) {$c$};
        \node[below of=iter] (loop) {};
        \path[->] (A) edge (c);
        \draw (A) to (na);
        \draw[dashed,->] (iter) to (A);
        \draw[dashed,->] (A) to (B);
        \draw[dashed] (B) to (loop);
        \draw[dashed,->] (loop) to (iter);
    \end{tikzpicture}
     &
     $\xrightarrow{\piout{c}{1!}}$
    &
     \begin{tikzpicture}[node distance=30pt]
        \node[circle,draw] (iter) {\phantom{$\circ$}};
        \draw node[above left=4pt] at (iter) {$*$};
        \node[circle,draw,right of=iter] (A) {\phantom{$\bullet$}};
        \node[circle,draw,below of=A,double] (B) {$\circ$};
        \node[rectangle,draw,right of=A] (na) {$1!$};
         \draw node[right=5pt] at (na) {${\scriptstyle \nu a}$};
        \node[rectangle,draw, below of=na] (c) {$c$};
        \node[below of=iter] (loop) {};
        \path[->] (A) edge (c);
        \draw (A) to (na);
        \draw[dashed,->] (iter) to (A);
        \draw[dashed,->] (A) to (B);
        \draw[dashed] (B) to (loop);
        \draw[dashed,->] (loop) to (iter);
    \end{tikzpicture}
     \\

 & $0\vdash \redex{*}[\piout{c}{\nu a}0]$ &

 & $0\vdash *[\redex{\piout{c}{\nu a}}0]$ &

 & $1\vdash *[\piout{c}{\nu a \sep 1!}\redex{0}]$ 
\\
  $\xrightarrow{\varepsilon}$ & 

     \begin{tikzpicture}[node distance=30pt]
        \node[circle,draw] (iter) {$\circ$};
        \draw node[above left=4pt] at (iter) {$*$};
        \node[circle,draw,right of=iter] (A) {\phantom{$\circ$}};
        \node[circle,draw,below of=A,double] (B) {\phantom{$\circ$}};
        \node[rectangle,draw,right of=A] (na) {$\nu a$};
        \node[rectangle,draw, below of=na] (c) {$c$};
        \node[below of=iter] (loop) {};
        \path[->] (A) edge (c);
        \draw (A) to (na);
        \draw[dashed,->] (iter) to (A);
        \draw[dashed,->] (A) to (B);
        \draw[dashed] (B) to (loop);
        \draw[dashed,->] (loop) to (iter);
    \end{tikzpicture}
& $\ldots$ &
$\xrightarrow{\piout{c}{2!}} \quad \ldots \quad \xrightarrow{\piout{c}{3!}}\quad \ldots$ 
\\
 & $1\vdash \redex{*}[\piout{c}{\nu a}0]$ &

  \end{tabular}
\end{center}
\caption{\label{fig:ex:iterator} Example with an iterator 
(abridged) }
\end{figure}

It is possible to express non-terminating behaviors with $\pi$-graphs
using \emph{iterators}~\cite{DBLP:conf/icalp/BusiGZ04}. The example of
Figure~\ref{fig:ex:iterator} is a $\pi$-graph encoding a generator of
fresh names. The iterator place is denoted $*$, which is marked in the
first step. An iteration is started with an $\varepsilon$ transition, a
low-level \emph{normalization step}. In the term representation, each
state is attached to a \emph{clock}. As
in~\cite{DBLP:conf/sofsem/PeschanskiB09} we can use \emph{logical
  clocks} to generate names that are guaranteed fresh by construction.
Consider the second transition on the figure.  The redex is the output
of the private name $\nu a$ on the public channel $c$. The effect of
transmitting a private name over a public channel (a \emph{bound
  output}) must be recorded.  The box of the formerly private name
$\nu a$ is then instantiated with $1!$ which is the new identity of
the name.  The generated name is the current value of the clock plus
one suffixed by $!$ to mark the output (a suffix $?$ is used for fresh
inputs). It is guaranteed fresh by construction and, to ensure this, the
clock itself is incremented by one.  The observation is
\emph{recorded} as a transition labelled $\piout{c}{1!}$, and we reach
the terminating place $0$. The iterator is then reactivated, and
during this step the box $\nu a$ is reinitialized to its default value
$\nu a\sep \nu a$. This makes the name $\nu a$ \emph{locally private}
to the iterator. There are also global private names or
\emph{restrictions} as in CCS, denoted e.g. $\nu A, \nu B\ldots$ These
are not reinitialized before the start of new iterations. We are now
in the same state except the value of the clock was incremented by
one.  Thus, if we continue iterating the behavior, the recorded
observations will be $\piout{c}{2!}$, $\piout{c}{3!}$, etc.  resulting
in an infinite generation of distinct names.

The last example, cf. Figure~\ref{fig:match:ex}, illustrates the interpretation of the match prefix.
We study the evolution of the following behavior: $\piout{c}{\nu a}\piin{d}{\iota x}[\nu a=\iota x]P$.
To anticipate the semantics of Table~\ref{tab:semantics} (see page~\pageref{tab:semantics}), we also indicate the names of the inferred transitions in the Figure. 

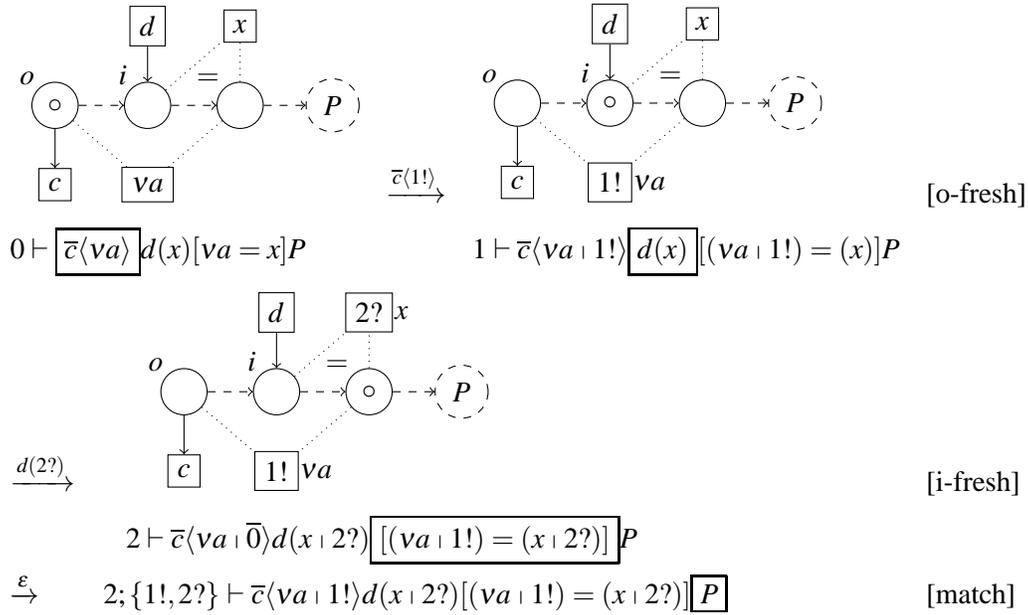
\begin{figure}[htb]
\begin{center}
\[\begin{array}{llll}
\begin{tikzpicture}[node distance=30pt]
    \node[circle,draw] (a) {$\circ$};
       \draw node[above left=4pt] at (a) {$o$};
    \node[circle,draw,right of=a,node distance=35pt] (b) {\phantom{$\circ$}};
       \draw node[above left=4pt] at (b) {$i$};
    \node[circle,draw,right of=b,node distance=35pt] (c) {\phantom{$\circ$}};
       \draw node[above left=4pt] at (c) {$=$};
    \node[circle,draw,dashed,right of=c,node distance=35pt] (d) {$P$};
    \draw[dashed,->] (a) to (b);
    \draw[dashed,->] (b) to (c);
    \draw[dashed,->] (c) to (d);

    \node[rectangle,draw,below of=a] (bc) {$c$};
    \node[rectangle,draw, below of=b] (na) {$\nu a$};
    \path[->] (a) edge (bc);
    \draw[dotted] (a) to (na);
    \node[rectangle,draw,above of=b] (bd) {$d$};
    \node[rectangle,draw, above of=c] (ix) {$\iota x$};
    \path[->] (bd) edge (b);
    \draw[dotted] (b) to (ix);
    \draw[dotted] (c) to (ix);
    \draw[dotted] (c) to (na);
\end{tikzpicture}
&
\xrightarrow{\piout{c}{1!}} 
&
\begin{tikzpicture}[node distance=30pt]
    \node[circle,draw] (a) {\phantom{$\circ$}};
       \draw node[above left=4pt] at (a) {$o$};
    \node[circle,draw,right of=a,node distance=35pt] (b) {$\circ$};
       \draw node[above left=4pt] at (b) {$i$};
    \node[circle,draw,right of=b,node distance=35pt] (c) {\phantom{$\circ$}};
       \draw node[above left=4pt] at (c) {$=$};
    \node[circle,draw,dashed,right of=c,node distance=35pt] (d) {$P$};
    \draw[dashed,->] (a) to (b);
    \draw[dashed,->] (b) to (c);
    \draw[dashed,->] (c) to (d);

    \node[rectangle,draw,below of=a] (bc) {$c$};
    \node[rectangle,draw, below of=b] (na) {$1!$};
       \draw node[right=5pt] at (na) {$\scriptsize \nu a$};
    \path[->] (a) edge (bc);
    \draw[dotted] (a) to (na);
    \node[rectangle,draw,above of=b] (bd) {$d$};
    \node[rectangle,draw, above of=c] (ix) {$\iota x$};
    \path[->] (bd) edge (b);
    \draw[dotted] (b) to (ix);
    \draw[dotted] (c) to (ix);
    \draw[dotted] (c) to (na);
\end{tikzpicture}
&
\textrm{[o-fresh]}\\[.2cm]
0 \vdash \redex{\piout{c}{\nu a}}\piin{d}{\iota x}[\nu a=\iota x]P  
&
&
1  \vdash \piout{c}{\nu a\sep 1!}\redex{\piin{d}{\iota x}}[(\nu a\sep 1!)=(\iota x)]P  
& 
\\[.3cm]
\multicolumn{3}{l}{
\xrightarrow{\piin{d}{2?}} 
\qquad
\begin{tikzpicture}[node distance=30pt]
    \node[circle,draw] (a) {\phantom{$\circ$}};
       \draw node[above left=4pt] at (a) {$o$};
    \node[circle,draw,right of=a,node distance=35pt] (b) {\phantom{$\circ$}};
       \draw node[above left=4pt] at (b) {$i$};
    \node[circle,draw,right of=b,node distance=35pt] (c) {$\circ$};
       \draw node[above left=4pt] at (c) {$=$};
    \node[circle,draw,dashed,right of=c,node distance=35pt] (d) {$P$};
    \draw[dashed,->] (a) to (b);
    \draw[dashed,->] (b) to (c);
    \draw[dashed,->] (c) to (d);

    \node[rectangle,draw,below of=a] (bc) {$c$};
    \node[rectangle,draw, below of=b] (na) {$1!$};
       \draw node[right=5pt] at (na) {$\scriptsize \nu a$};
    \path[->] (a) edge (bc);
    \draw[dotted] (a) to (na);
    \node[rectangle,draw,above of=b] (bd) {$d$};
    \node[rectangle,draw, above of=c] (ix) {$2?$};
       \draw node[right=6pt] at (ix) {$\scriptsize \iota x$};
    \path[->] (bd) edge (b);
    \draw[dotted] (b) to (ix);
    \draw[dotted] (c) to (ix);
    \draw[dotted] (c) to (na);
\end{tikzpicture}
}
&
\textrm{[i-fresh]}\\[.2cm]
\multicolumn{4}{l}{
\qquad \qquad
2  \vdash \piout{c}{\nu a\sep \overline{0}}\piin{d}{\iota x\sep 2?}\redex{[(\nu a\sep 1!)=(\iota x\sep 2?)]}P 
}
\\[.2cm]
\multicolumn{3}{l}{
\xrightarrow{\varepsilon} 
\qquad 
2;\{1!,2?\} \vdash \piout{c}{\nu a\sep 1!}\piin{d}{\iota x\sep 2?}[(\nu a\sep 1!)=(\iota x\sep 2?)]\redex{P} 
}
&
\textrm{[match]}
\\
\end{array}\]
\end{center}
\caption{\label{fig:match:ex} Example of match (abridged)}
\end{figure}

In the initial state, the logical clock value is 0. The first action is
an emission of the private name $\nu a$ over the public channel $c$.
This leads to the observation $\piout{c}{1!}$ and the clock value
becomes $1$. The second step is a reception from the public channel
$d$, the received name is selected fresh and it is denoted $2?$. The
clock value is once again incremented. Now a match is performed,
testing if the fresh names $1!$ and $2?$ can be made equal. The answer
is positive because the name $1!$ has been sent \emph{before} $2?$ is
received. The justification of this causal link is simply the
comparison of their respective clock value, i.e., $1<2$. To perform the
match we record in the context of the $\pi$-graph, together with the
clock, a (dynamic) partition of names wrt. equality. By default, all
the names are considered distinct and thus the partition only contains
singletons, which are left implicit for the sake of readability. In
the final state, the partition is refined so that the singletons
$\{1!\},\{2?\}$ are replaced by their union $\{1!,2?\}$. In this
context (and thus in the continuation $P$) the two names are
considered equal. Now, if we perform first the reception and then the
emission, a causal link should not exist and we thus expect the two
names \emph{cannot} be made equivalent.
\[
0   \vdash \redex{\piin{d}{\iota x}}\piout{c}{\nu a}[\nu a=\iota x]P \quad
\xrightarrow{\piin{d}{1?}} \quad \cdots \quad
\xrightarrow{\piout{c}{2!}} \quad
2 \vdash \piin{d}{\iota x\sep 1?}\piout{c}{\nu a\sep 2!}\redex{[(\nu a\sep 2!)=(\iota x\sep 1?)]}P 
\]
The match fails because the names $1?$ and $2!$ cannot be equated, which is because $2<1$ does \emph{not} hold. 

There are two distinct abstraction levels where the properties of the $\pi$-graphs can be discussed:  the process algebra level and the lower-level of the underlying graph model. We now give the basic definitions of the graph model.

\begin{definition} \label{def:names}
The set of names is $\mathcal{N}\defs \mathcal{N}_f \uplus \mathcal{N}_b\uplus \mathcal{N}_r \uplus \mathcal{N}_p \uplus \mathcal{N}_o \uplus \mathcal{N}_i$ with:

$\left [\begin{array}{l}
\mathcal{N}_f  \text{ the set of free names } a,b,\ldots \\
\mathcal{N}_b  \text{ the set of binder names } \iota x,\iota y,\ldots \\
\mathcal{N}_r  \text{ the set of restrictions } {\nu}A,{\nu}B,\ldots \\
\mathcal{N}_p \text{ the set of private names } {\nu}a,{\nu}b,\ldots \\
\mathcal{N}_o\defs \{ n! \mid n\in\mathbb{N}\} \text{ the set of fresh outputs}\\
\mathcal{N}_i\defs \{ n? \mid n\in\mathbb{N}\} \text{ the set of fresh inputs}
\end{array}\right .$

We also define $\fun{Priv}\defs \mathcal{N}_r \cup \mathcal{N}_p$ (private names), $\fun{Pub}\defs\mathcal{N}\setminus \fun{Priv}$ (public names) and $\fun{Stat}\defs  \mathcal{N}\setminus ( \mathcal{N}_o\cup \mathcal{N}_i)$ (static names) 
\end{definition}

\begin{definition} \label{def:pigraph}
A configuration is a tuple $\pi\defs \langle \kappa, \gamma, P,\fun{pt},B,\fun{bn},\fun{data},\fun{in},\fun{out},\fun{ctl},M,I \rangle$ with 
\begin{itemize}
\item $\kappa \in \mathcal{K}$ a \emph{clock} value (see below),
\item $\gamma \subseteq \powerset{ \mathcal{N}_f\cup\mathcal{N}_i\cup\mathcal{N}_o}$ a partition of names, 
\item $P$ a finite abstract set of places,
\item $\fun{pt} : P \rightarrow \{0,\tau,i,o,=,\sum,\prod,*\}$ the place types,
\item $B$ a finite abstract set of boxes,
\item $\fun{bn} : B \rightarrow \fun{Stat}$ an injective function for box names,
\item $\fun{data},\fun{in},\fun{out} : P \rightarrow B \cup \{\bot\}$ the data, input and output links,
\item $\fun{ctl} : P \rightarrow \powerset{P}$ the control links, 
\item $M : P\rightarrow \{\circ,\emptyset\}$ a marking function ($\circ$ redex, $\emptyset$ empty mark),
\item $I : B\rightarrow \mathcal{N}$  a box instantiation function.
\end{itemize}
\end{definition}

In Definition~\ref{def:pigraph}, $\kappa$, $\gamma$, $M$ and $I$ are the only dynamic elements;
they will evolve through the application of the semantic rules, cf. Table~\ref{tab:semantics}.
Initially, the partition $\gamma$ contains only the singleton subsets of the infinite set 
$\mathcal{N}_f \cup\mathcal{N}_i\cup \mathcal{N}_o$ of names, $I$ is $\fun{bn}$ and the marking $M$ corresponds to the $\circ$-marking of the initial place of each iterator.
An initial $\pi$-graph is a configuration that is well-formed according to the syntax rules of Table~\ref{tab:syntax}.
A $\pi$-graph is a configuration that is both well-formed and reachable from an initial one by application of the semantic rules.
Only well-formed $\pi$-graphs will be considered in the following.
In order to keep the notations compact, we shall classically omit the singleton sets of a partition $\gamma$ (hence, initially, $\gamma=\emptyset$).

\begin{table}[p]
\begin{center}
\begin{tabular}{lcr}
\bf{Prefixes} $p$ ::= \\
\begin{tikzpicture}
\draw (-1,0) node[circle,fill=white,draw] (p0) {\phantom{$\bullet$}};
\draw (-1,.5) node (p0p) {$\scriptstyle\tau$};
\draw (0,0) node[circle,fill=white,draw,dashed] (p1) {\phantom{$\bullet$}};
\draw[dashed,->] (p0) to (p1);
\end{tikzpicture}
&
$\mid$
&
\begin{tikzpicture}
\draw (-2,.3) node[rectangle,draw] (a) {$\delta$};
\draw (-2.4,.3) node (a1) {$\scriptstyle\Delta$};
\draw (-2,-.3) node[rectangle,draw] (c) {$\varphi$};
\draw (-2.4,-.3) node (c1) {$\scriptstyle\Phi$};
\draw (-1,0) node[circle,fill=white,draw] (p0) {\phantom{$\bullet$}};
\draw (-1,.5) node (p0p) {$\scriptstyle o$};
\path[->] (p0) edge (c);
\path[dotted] (a) edge (p0);
\draw (0,0) node[circle,fill=white,draw,dashed] (p1) {\phantom{$\bullet$}};
\draw[dashed,->] (p0) to (p1);
\end{tikzpicture} \\[.2cm]
Silent  ${\tau}$  & & Output $\piout{\Phi\sep\varphi}{\Delta\sep\delta}$ \\[.3cm]
\begin{tikzpicture}
\draw (-2,.3) node[rectangle,draw] (a) {$\iota x$};
\draw (-2,-.3) node[rectangle,draw] (c) {$\varphi$};
\draw (-2.4,-.3) node (c1) {$\scriptstyle\Phi$};
\draw (-1,0) node[circle,fill=white,draw] (p0) {\phantom{$\bullet$}};
\path[->] (c) edge (p0);
\path[dotted] (a) edge (p0);
\draw (-1,.5) node (p0p) {$\scriptstyle i$};
\draw (0,0) node[circle,fill=white,draw,dashed] (p1) {\phantom{$\bullet$}};
\draw[dashed,->] (p0) to (p1);
\end{tikzpicture} 
&
$\mid$
&
\begin{tikzpicture}
\draw (-2,.3) node[rectangle,draw] (a) {$\delta$};
\draw (-2.4,.3) node (a1) {$\scriptstyle\Delta$};
\draw (-2,-.3) node[rectangle,draw] (c) {$\varphi$};
\draw (-2.4,-.3) node (c1) {$\scriptstyle\Phi$};
\draw (-1,0) node[circle,fill=white,draw] (p0) {\phantom{$\large\bullet$}};
\draw (-1,.5) node (p0p) {$\scriptstyle =$};
\path[dotted] (p0) edge (c);
\path[dotted] (a) edge (p0);
\draw (0,0) node[circle,fill=white,draw,dashed] (p1) {\phantom{$\bullet$}};
\draw[dashed,->] (p0) to (p1);
\end{tikzpicture} \\[.2cm]
Input  $\piin{\Phi\sep\varphi}{\iota x}$ && Match  $\pim{\Phi\sep\varphi}{\Delta\sep\delta}$ \\[.3cm]
\begin{tikzpicture}
\draw (-1,1) node[circle,fill=white,draw] (p0) {\phantom{$\bullet$}};
\draw (-1.4,1.4) node (p0p) {$\sum$};
\draw (.75,1.1) node (dots) {$\vdots$};
\draw[densely dashed,thick,rounded corners] (0,1.3) rectangle (1.5,2.1);
\draw (0,1.7) node[circle,fill=white,draw] (p1a) {\phantom{$\bullet$}};
\draw (.75,1.7) node (G) {$P_{1}$};
\draw (1.5,1.7) node[circle,fill=white,draw,double] (p2a) {\phantom{$\bullet$}};
\draw[densely dashed,thick,rounded corners] (0,-.1) rectangle (1.5,.7);
\draw (0,0.3) node[circle,fill=white,draw] (p1c) {\phantom{$\bullet$}};
\draw (.75,0.3) node (G) {$P_{n}$};
\draw (1.5,0.3) node[circle,fill=white,draw,double] (p2c) {\phantom{$\bullet$}};
\draw (2.5,1) node[circle,fill=white,draw,dashed] (he) {\phantom{$\bullet$}};
\draw[dashed,->] (p0) to (p1a);
\draw[dashed,->] (p0) to (p1c);
\draw[dashed,->] (p2a) to (he);
\draw[dashed,->] (p2c) to (he);
\end{tikzpicture} 
&
$\mid$
&
\begin{tikzpicture}
\draw (-1,1) node[circle,fill=white,draw] (p0) {\phantom{$\bullet$}};
\draw (-1.4,1.4) node (p0p) {$\prod$};
\draw (.75,1.1) node (dots) {$\vdots$};
\draw[densely dashed,thick,rounded corners] (0,1.3) rectangle (1.5,2.1);
\draw (0,1.7) node[circle,fill=white,draw] (p1a) {\phantom{$\bullet$}};
\draw (.75,1.7) node (G) {$P_{1}$};
\draw (1.5,1.7) node[circle,fill=white,draw,double] (p2a) {\phantom{$\bullet$}};
\draw[densely dashed,thick,rounded corners] (0,-.1) rectangle (1.5,.7);
\draw (0,0.3) node[circle,fill=white,draw] (p1c) {\phantom{$\bullet$}};
\draw (.75,0.3) node (G) {$P_{n}$};
\draw (1.5,0.3) node[circle,fill=white,draw,double] (p2c) {\phantom{$\bullet$}};
\draw (2.5,1) node[circle,fill=white,draw,dashed] (he) {\phantom{$\bullet$}};
\draw[dashed,->] (p0) to (p1a);
\draw[dashed,->] (p0) to (p1c);
\draw[dashed,->] (p2a) to (he);
\draw[dashed,->] (p2c) to (he);
\end{tikzpicture} \\[.2cm]
Choice  $\sum[P_{1} + \ldots + P_{n}] \quad (n>1) $ && Parallel  $\prod[P_{1} \parallel \ldots \parallel P_{n}] \quad (n>1)$ \\[.4cm]
{\bf Processes} $P$ ::=  \\[.2cm]
 \begin{tikzpicture}
\draw (-1,0) node[circle,fill=white,draw] (p0) {\phantom{$\bullet$}};
\draw[densely dotted,rounded corners] (-1,-.4) rectangle (.5,.4);
\draw (-1,0) node[circle,fill=white,draw] (p1a) {\phantom{$\bullet$}};
\draw (-.25,0) node (G1) {$p$};
\draw (.5,0) node[circle,fill=white,draw,double] (p1c) {\phantom{$\bullet$}};
%
\end{tikzpicture} 
& 
$\mid$
&
\begin{tikzpicture}
\draw (-1,0) node[circle,fill=white,draw] (p0) {\phantom{$\bullet$}};
\draw[densely dotted,rounded corners] (-1,-.4) rectangle (.5,.4);
\draw (-1,0) node[circle,fill=white,draw] (p1a) {\phantom{$\bullet$}};
\draw (-.25,0) node (G1) {$p$};
\draw[densely dashed,thick,rounded corners] (.5,-.4) rectangle (2,.4);
\draw (.5,0) node[circle,fill=white,draw] (p1c) {\phantom{$\bullet$}};
\draw (1.25,0) node (G2) {$P$};
\draw (2,0) node[circle,fill=white,draw,double] (p2c) {\phantom{$\bullet$}};
\end{tikzpicture} \\[.2cm]
Termination  $p0$ $\quad$ ($p\neq$  match)  & & Prefixed process $pP$ \\[.3cm]
\multicolumn{3}{l}{\bf{Iterator} $I$ ::=  }\\
 \multicolumn{3}{c}{
\begin{tikzpicture}
\draw (-2,1.3) node[rectangle,draw] (a1) {${\nu}a_1$};
\draw (-2,1) node (dots) {$\vdots$};
\draw (-2,0.5) node[rectangle,draw] (an) {${\nu}a_n$};
\draw (-1,1.8) node[rectangle,draw] (x1) {${\iota}x_1$};
\draw (-0.3,1.8) node (dots) {$\ldots$};
\draw (.4,1.8) node[rectangle,draw] (xm) {${\iota}x_m$};
\draw[densely dashed,thick,rounded corners] (0,0.6) rectangle (1.5,1.4);
\draw (-1.3,1.3) node (it) {*};
\draw (-1,1) node[circle,fill=white,draw] (p0) {\phantom{$\bullet$}}; 
\draw (0,1) node[circle,fill=white,draw] (p1) {\phantom{$\bullet$}};
\draw (1.5,1) node[circle,fill=white,draw,double] (p2) {\phantom{$\bullet$}};
\draw (.75,1) node (G) {$P$};
\draw (2.0,1) node (i1) {$\hspace*{-5pt}\neg$};
\draw (2.0,0.2) node (i2) {.};
\draw (-1,0.2) node (i3) {.};
\draw[dashed] (p2) to (i1);
\draw[dashed] (i1) to (i2);
\draw[dashed] (i2) to (i3);
\draw[dashed,->] (i3) to (p0);
\draw[dashed,->] (p0) to (p1);
\draw[dotted] (a1) to (p0);
\draw[dotted] (an) to (p0);
\draw[dotted] (x1) to (p0);
\draw[dotted] (xm) to (p0);
\end{tikzpicture} 
}\\
\multicolumn{3}{c}{
$*[({\nu}a_1,\ldots,{\nu}a_n)({\iota}x_1,\ldots,{\iota}x_m)~P]$ 
} \\[.4cm]
\multicolumn{3}{l}{\bf{Graph} $\pi$ ::=  $(a_1\ldots a_i) (\nu A_1,\ldots,\nu A_j)~[I_1 \parallel \ldots \parallel I_k] \quad$  ($k\geq 1$)} \\

\end{tabular}
\end{center}
\caption{\label{tab:syntax} Syntax}
\end{table}

A graph declares a set of free names (in $\mathcal{N}_f$), denoted $(a_1,\ldots,a_i)$, a set of global restrictions  (in $\mathcal{N}_r$), 
denoted $(\nu A_1,\ldots, \nu A_j)$
and a parallel composition of $k$ iterators, 
$k\geq 1$. 
An iterator declares a set of (locally) private names  (in $\mathcal{N}_P$), denoted $(\nu a_1,\ldots,\nu a_n)$, a set of
binder names  (in $\mathcal{N}_B$), denoted $(\iota x_1,\ldots,\iota x_m)$, and an iterated process $P$.
The place labeled $*$ is the initial place of the iterator. 
A process $P$ is a non-empty sequence of prefixes $p$ terminated by $0$; the latter corresponds to a unique place, of type $0$, represented with a double border.
Each prefix has (see Table~\ref{tab:syntax}) a unique terminating 
place, represented with a dashed border, which will be used to 
glue the prefixes together, and a unique initial place.  
A silent prefix has no box and an initial place labeled $\tau$.
An output prefix $\piout{\Phi\sep\varphi}{\Delta\sep\delta}$, 
whose initial place is labeled $o$, allows to emit 
a formal name $\Delta$, instantiated by $\delta$, 
on a channel with a formal name $\Phi$, instantiated by $\phi$.
This is indicated by a data (dotted) and an output (plain) link, respectively.
Each formal name is represented by a box with the instantiated name inside. 
We systematically omit box identities if they are the same as their instantiation.
Initially, it is in the form $\piout{\Phi\sep\Phi}{\Delta\sep\Delta}$, usually condensed in 
$\piout{\Phi}{\Delta}$, and in the graphical representation, the identity of the nodes is omitted 
if it is considered irrelevant, or may be inferred by the context.
An input prefix $\piin{\Phi\sep\varphi}{\iota x}$, 
whose initial place is labeled $i$, allows to receive 
an instantiation for the formal name $\iota x$ 
on a channel with a formal name $\Phi$, instantiated by $\phi$.
This is indicated by a data and input link, respectively.
A match prefix $\pim{\Phi\sep\varphi}{\Delta\sep\delta}$,
whose initial place is labeled ${=}$, allows to identify 
a formal name $\Delta$, instantiated by $\delta$, 
with a formal name $\Phi$, instantiated by $\varphi$.
This is indicated by two data links.
A choice prefix $\sum[P_{1} + \ldots + P_{n}]$ allows to choose one out of several processes $P_1$ to $P_n$;
it starts with a place labeled $\sum$ connected to the starting place of each of those processes, and each terminating place of a process is connected to the terminating place of the choice prefix.
A parallel prefix $\prod[P_{1} + \ldots + P_{n}]$ allows to activate simultaneously all the processes $P_1$ to $P_n$;
it starts with a place labeled $\prod$ connected to the starting place of each of those processes, and each terminating place of a process is connected to the terminating place of the parallel prefix.

A \emph{clock model} is a type $\mathcal{K}$ associated to a set of operations with the following signatures:
$\fun{init}:\mathcal{K};
\fun{in}, \fun{out}: \mathcal{K} \rightarrow \mathcal{K};
\fun{next_i},\fun{next_o}: \mathcal{K} \rightarrow \mathbb{N};
\prec: \mathcal{K} \times\mathcal{N}_o \times \mathcal{N}_i \rightarrow \mathbb{B}$. 
In the semantics (cf. the next Section), it is assumed that every transition path 
starts with the initial clock value $\fun{init}$. The function $\fun{in}$ (resp. $\fun{out}$) is used
to update the clock when an input (resp. an output) is performed. The identity of the fresh names is
 generated with $\fun{next_i}$ (fresh input) and $\fun{next_o}$ (fresh output). The read-write causality ordering
  is expressed by the $\prec$ relation. A triplet $(\kappa,n!,m?)\in \prec$ is denoted $n!\prec_\kappa m?$.

In the following we will be interested in a freshness property of a clock model.
 
\begin{definition}
\label{def:freshness}
Let $\pi$ be a graph with clock $\kappa$ and instantiation $I$, then $\pi$ satisfies the freshness constraint if:
$\fun{next_o}(\kappa)!\not\in \fun{cod}(I)\land \fun{next_i}(\kappa)?\not\in \fun{cod}(I)$.
\end{definition}

Notice that any initial graph satisfies the freshness constraint since 
$\fun{cod}(I)\cap (\mathcal{N}_i \cup \mathcal{N}_o)=\emptyset$. 

A clock model satisfies the \emph{freshness constraint} if for any $\pi$ reachable 
from an initial one using the evolution rules described in the next section, 
the freshness constraint is preserved.

The simplest model of \emph{logical clocks} is such that $\mathcal{K} = \mathbb{N}$ with: \\
$\fun{init} = 0,\fun{out}(\kappa) = \fun{next_o}(\kappa),\fun{in}(\kappa) = \fun{next_i}(\kappa),\fun{next_o}(\kappa) =\fun{next_i}(\kappa) = \kappa+1$ and $n! \prec_\kappa m? \textrm{ iff } n < m$. Such logical clocks trivially satisfy the freshness constraint.

\section{Operational semantics}
\label{sec:semantics}
 
\begin{table}[p]
\begin{center}

\begin{tabular}{ll}
[silent] & $\kappa;\gamma\vdash \redex{\tau}P \xrightarrow{\tau} \kappa;\gamma\vdash{\tau}\redex{P}$
\\[0.5cm]
[out] & $\kappa;\gamma\vdash \redex{\piout{\Phi\sep\varphi}{\Delta\sep\delta}}P \xrightarrow{\piout{\varphi}{\delta}} \kappa;\gamma\vdash \piout{\Phi\sep\varphi}{\Delta\sep\delta}\redex{P}$ \quad if $\Phi,\Delta\in\fun{Pub}$ 
\\[0.5cm]
[o-fresh] & $\kappa;\gamma\vdash \redex{\piout{\Phi\sep\varphi}{{\nu}\alpha}}P \xrightarrow{\piout{\varphi}{\fun{next_o}(\kappa)!}} \fun{out}(\kappa);\gamma\vdash 
\piout{\Phi\sep\varphi}{{\nu}\alpha\sep{\fun{next_o}(\kappa)}!}\redex{P}$ \\ & if $\Phi\in\fun{Pub},~\nu\alpha\in\fun{Priv}$
\\[0.5cm]
%
%
[i-fresh] & $\kappa;\gamma\vdash \redex{\piin{\Phi\sep\varphi}{\iota x}}P \xrightarrow{\piin{\varphi}{\fun{next_i}(\kappa)?}} \fun{in}(\kappa);\gamma\vdash \piin{\Phi\sep\varphi}{\iota x\sep \fun{next_i}(\kappa)?}\redex{P}$ \quad if $\Phi\in\fun{Pub}$ 
\\[0.5cm]
%
%
[match] & $\kappa;\gamma\vdash \redex{\pim{\Phi\sep\varphi}{\Phi'\sep\varphi'}}P \xrightarrow{\varepsilon} \kappa;\gamma_{\lhd \varphi=\varphi'} \vdash \pim{\Phi\sep\varphi}{\Phi'\sep\varphi'}\redex{P}$  $\quad\textrm{ if } \varphi\compat{\gamma}_\kappa\varphi'$ 
\\[0.5cm] 
[sync] & $\kappa;\gamma\vdash \redex{\piout{\Phi\sep\varphi}{\Delta\sep\delta}}P \parallel \redex{\piin{\Phi'\sep\varphi'}{\iota x\sep\iota x}}Q$ \\ & $\xrightarrow{\tau} \kappa;\gamma_{\lhd \varphi=\varphi'} \vdash \piout{\Phi\sep\varphi}{\Delta\sep\delta}\redex{P} \parallel \piin{\Phi'\sep\varphi'}{\iota x\sep\delta}\redex{Q}$ $\quad \textrm{ if } \varphi\compat{\gamma}_\kappa\varphi'$ 
\\[0.5cm] 
[sum] & $\kappa;\gamma\vdash \redex{\sum}[P_{1} + \ldots  +  P_{i} +  \ldots +  P_{n}]Q$ 
$\xrightarrow{\mu} \kappa';\gamma'\vdash \sum[P_{1} + \ldots +  \mathcal{P}_i + \ldots +  P_{n}]Q$ \\
&  if ${\exists}\mu\neq\varepsilon,~\kappa;\gamma\vdash \redex{P_{i}} \xrightarrow{\varepsilon^*\mu} \kappa';\gamma' \vdash \mathcal{P}_i$
\\[0.5cm] 
[sum$_0$] & $\kappa;\gamma\vdash \sum[P_{1} + \ldots  +  P _{i}{\redex{0}} +  \ldots +  P_{n}]Q$  
$\xrightarrow{\varepsilon} \kappa;\gamma\vdash \sum[P_{1} + \ldots +  P _{i}{0} + \ldots +  P_{n}]\redex{Q}$ 
\\[0.5cm] 
[par] & $\kappa;\gamma\vdash \redex{\prod}[P_{1}\parallel\ldots \parallel P_{i}\parallel \ldots\parallel P_{n}]Q$
$\xrightarrow{\varepsilon} \kappa;\gamma\vdash \prod[\redex{P_{1}}\parallel\ldots\parallel\redex{P_{i}}\parallel\ldots\parallel \redex{P_{n}}]Q$ 
\\[0.5cm]
[par$_0$] & $\kappa;\gamma\vdash \prod[P_{1}{\redex{0}}\parallel\ldots \parallel P_{i}{\redex{0}}\parallel \ldots\parallel P_{n}{\redex{0}}]Q$ 
$\xrightarrow{\varepsilon} \kappa;\gamma\vdash \prod[P _{1}{0}\parallel\ldots\parallel P_{i}{0}\parallel\ldots\parallel P_{n}{0}]\redex{Q}$ 
\\[0.5cm] 
[iter] &  $\kappa;\gamma\vdash \redex{*}[({\nu}a_1\sep\delta_1),\ldots,(\nu a_n\sep\delta_n)|
					({\iota}x_1\sep\varphi_1),\ldots,(\iota x_m\sep\varphi_m)~P]$ \\
&  $\xrightarrow{\varepsilon} \kappa;\gamma\vdash *[({\nu}a_1\sep \delta_1),\ldots,(\nu a_n\sep \delta_n)|
					({\iota}x_1\sep\varphi_1),\ldots,(\iota x_m\sep\varphi_m)~\redex{P}]$
\\[0.5cm]
[iter$_0$] &  $\kappa;\gamma\vdash *[({\nu}a_1\sep\delta_1),\ldots,(\nu a_n\sep\delta_n)|
					({\iota}x_1\sep\varphi_1),\ldots,(\iota x_m\sep\varphi_m)~P{\redex{0}}]$ \\
&  $\xrightarrow{\varepsilon} \kappa;\gamma\vdash \redex{*}[({\nu}a_1),\ldots,(\nu a_n)|
					({\iota}x_1),\ldots,(\iota x_m)~P0]$
\end{tabular}
\end{center}
\caption{\label{tab:semantics} The operational semantics rules.}
\end{table}


\begin{sloppypar}
  The operational semantics for the $\pi$-graphs provide the meaning
  of the one-step transition relation $\xrightarrow{.}$
  (the dot symbol denotes an arbitrary label).  The rules of
  Table~\ref{tab:semantics} describe the local updates of a global graph
  $\pi \defs \langle \kappa,
  P,\fun{pt},B,\fun{bn},\fun{data},\fun{in},\fun{out},\fun{ctl},M,I
  \rangle$.  Most rules are of the form
  \[
  \kappa;\gamma \vdash pP \quad \xrightarrow{\alpha} \quad \kappa' ;\gamma'\vdash
  p'P'
  \] 
  where $\kappa;\gamma$ is the \emph{global context} of the rule.  The
  left-hand side (LHS) is a pattern describing a \emph{local context} composed of a prefix $p$
  and its continuation $P$.  The right-hand side (RHS) is an updated
  version of the local context. The rule is applicable if a subgraph
  of $\pi$ matches the LHS. In this case a (global) transition
  labelled $\alpha$ occurs and the matched subgraph in $\pi$ is
  updated according to the RHS. The global context of $\pi$ may also be
  updated. For example, the LHS of the [silent] rule identifies a
  sub-graph of $\pi$ consisting of a place $p\in P$ such that
  $\fun{pt}(p)=\tau$ and $M(p)=\circ$, followed by its
  continuation\footnote{According to the syntax
    (cf. Table~\ref{tab:syntax}), the continuation of a prefix
    is either a place $0$ or the initial place of the next
    prefix in the sequence.}. The RHS of the rule describes the
  next state $\pi'$ with a global context unchanged. The local context
  is updated so that the token in $p$ is passed to the initial place
  $q$ of the continuation, i.e., in the image $\pi'$, we have $M'(p)=\emptyset$ and
  $M'(q)=\circ$. We put a frame around a whole process to denote the presence of 
  a token $\circ$ in its initial
  place. The
  inferred transition carries the label $\tau$, which corresponds to a
  silent transition.
\end{sloppypar}

The [par] rule is similar to the silent step except that the token is
replicated for all the continuation places, simulating the fork of
parallel processes. The latter works in conjunction with the [par$_0$]
rule, which waits for all the forked processes to terminate before
passing the token to the continuation place. We use a $0$ suffix to make
explicit the termination place of the process when required.  The iterators are
operated in a similar way using the [iter] and [iter$_0$] rules. As
illustrated in the example of Figure~\ref{fig:ex:iterator}, each box
  $b$ for private or binder names ($\fun{bn}(b)\in\mathcal{N}_p\cup\mathcal{N}_b$) is
  reinitialized ($I(b)=\fun{bn}(b)$) at the end of each iteration.

The choice operator requires as in the $\pi$-calculus to play ``one move in
 advance'': the [sum] rule applies if we can follow a branch of
the choice such that at some point an observation can be made, possibly after an
arbitrary - but \emph{finite} - sequence of $\varepsilon$-transitions (cf. Lemma \ref{lemma:epsilon:term}).

The communication rules are critical components of the semantics. The
[out] rule applies when a process emits a public value using a public
channel (i.e., in set $\fun{Pub}$).  The effect of the rule is to
produce a transition with the observation as a label. 
The LHS of the [o-fresh] rule matches the
emission of a private name over a public channel. As explained in the
example of Figure~\ref{fig:ex:iterator}, the principle is to generate a
name that is guaranteed fresh by construction. This is obtained by
taking the next value of the current clock, which gives
$\fun{next_o}(\kappa)!$. To preserve the freshness constraint (cf.
Definition~\ref{def:freshness}), the clock itself is updated. For
example, if $\kappa$ is a logical clock assigned to the
value $3$ then the generated fresh name is denoted $4!$ (fresh by construction) and the clock
evolves to the value $4$. 

The rule for input is quite similar to the output ones. When a name is received from the environment, the [i-fresh] rule 
generates a fresh identity $\fun{next_i}(\kappa)?$ for it and records the observation. 

The rule [sync] is for a communication taking place internally in a
$\pi$-graph. The LHS of the rule matches two subgraphs in distinct
parallel processes , one is an output prefix with a $\circ$-token and the other
one a corresponding input also with a $\circ$-token (and both with their
respective continuations). The rule can be triggered either if the two
processes belong to different parallel branches of execution within the same
iterator, or if they are components of two distinct iterators. In both
cases, the effect of the rule is the same: the tokens are passed to
the respective continuations and the box of the input prefix is
instantiated with the emitted value. Similarly to late congruence
for the $\pi$-calculus, the communication can be triggered if the
partners potentially agree on the name of the channels. The communication
rule thus ``incorporates'' the semantics of the match prefix.

The matching of names is a central aspect of the proposed
semantics. It is indeed required in both the [match] and [sync]
rules. As illustrated by the symbolic semantics
of~\cite{DBLP:journals/iandc/BorealeN96}, matching in the
$\pi$-calculus is non-trivial because equality on names is dynamic,
i.e. two distinct names $a,b$ can be made equal through a match, under
certain conditions. In this work, the conditions we use relate to a
form of read-write causality~\cite{DBLP:conf/icalp/DeganoP95}. Instead of just comparing names, the equality relation on names can be dynamically refined by
updating the partition $\gamma$ (cf. the last example of
Section~\ref{sec:language}). The condition for the matching of two
names $\delta,\delta'$ under some clock $\kappa$ is denoted
$\delta\equi_\kappa \delta'$.

\begin{definition}\label{def:compat}
$\equi_\kappa$ is the smallest reflexive and symmetric binary relation on $\mathcal{N}$ such that
$\delta \equi_\kappa \delta'  \textrm{ if } \\
(\delta,\delta'\in \mathcal{N}_f \cup \mathcal{N}_i)
\lor (\delta=n! \in  \mathcal{N}_o \land\delta'=m?\in  \mathcal{N}_i \land n\prec_{\kappa} m)
$ 
\end{definition}

If $\delta$ is a free (public) name ($\delta\in\mathcal{N}_f$), there
are two possibilities for $\delta'$ to match it: either it is also a
free name or it belongs to the set of fresh input names. Indeed, we
may always receive a public name from the environment. If both 
names correspond to (fresh) inputs, they may also be equated. The
most delicate case is when $\delta$ is a fresh output and $\delta'$ a
fresh input.  As illustrated in Section \ref{sec:language}), the names can only be equated if the input is
causally dependent on the output.

The partition of names $\gamma$ can be \emph{refined} 
by a new equality $\delta=\delta'$ using the notation $\gamma_{\lhd \delta=\delta'}$,
if $\delta$ and $\delta'$ are \emph{compatible}, which is denoted $\delta \compat{\gamma}_\kappa \delta'$.

\begin{definition}\label{def:partition-update}
Let $\gamma$ be a partition of names, $\kappa$ a clock and $\delta$ and $\delta'$ names.
\begin{enumerate}
	\item $\delta \compat{\gamma}_\kappa \delta'$ iff 
$\forall n\in [\delta]_{\gamma}, \forall m\in [\delta']_{\gamma}\colon n\equi_\kappa m$,
\item $\gamma_{\lhd \delta=\delta'}\defs (\gamma\setminus \{[\delta]_{\gamma},[\delta']_{\gamma}\}) \cup
		\{ [\delta]_{\gamma} \cup [\delta']_{\gamma}\}$ if $\delta \compat{\gamma}_\kappa \delta'$.
\end{enumerate}
\end{definition}

This updates the partition so that a new equality holds, but only if the two names can \emph{actually} be made equal. The notation $[\delta]_\gamma$ denotes the equivalence class of $\delta$ in the relation $\gamma$. The following proposition plays a role in the finiteness results of Section~\ref{sec:results}.

\begin{proposition} \label{prop:no-two-output-equivalence}
Let $\pi$ a graph with partition $\gamma$. For any $E\in \gamma$, $n!\in E \implies E\setminus \{ n! \} \subset \mathcal{N}_i$
\end{proposition}

\begin{proof}
This simply says that fresh output names can only be made equal with (fresh) input names, which is a direct consequence of Definition~\ref{def:compat}, Definition~\ref{def:partition-update}, the way it is used in the operational semantics, and the fact that initially all classes are singletons.
\end{proof}

A central property for the remaining developments is that there is a finite bound on the length of $\varepsilon$-sequences involved in the semantics.
To demonstrate this result, we first need to introduce the notion of \emph{full path} of a (terminated) process. 

\begin{definition}
Let $P0$ be a process. A full path of it is a sequence $\sigma$ of transitions leading from $\redex{P}0$ to $P\redex{0}$.
\end{definition}

\begin{lemma} \label{lemma:full-epsilon-path}
No full path may be an $\varepsilon$-sequence.
\end{lemma}

\begin{proof}
The demonstration is by a simple structural induction on the syntax. 
First, the termination $0$ cannot be preceded by a match. 
Moreover, the property holds by induction for the parallel and sum sub-processes, which are the only prefixes able to generate an $\varepsilon$ at the end of a full path. 
\end{proof} 

\begin{lemma} \label{lemma:epsilon:term}
For any graph $\pi$, there is a finite bound on the length of the $\varepsilon$-sequences it may generate.
\end{lemma}

\begin{proof}
  First, note that we may neglect synchronisations, since they yield $\tau$-transitions and not $\varepsilon$-ones.
 If there are several iterators, we may interleave their longest $\varepsilon$-sequences and a bound is provided by the sum of the bounds of each component iterator.
  For any iterator $*[P0]$, we know from Lemma~\ref{lemma:full-epsilon-path} that no $\varepsilon$-sequence of $P$ may be both initial and terminal in a full path it generates; hence, 
besides the $\varepsilon$-sequences generated by $P$, we may also have a terminal one followed by [iter$_0$], followed by [iter], followed by an initial one (and we may not loop indefinitely on full $\epsilon$-paths), so that a bound is provided by twice the bound\footnote{Better bounds could be obtained by separately evaluating bounds for initial, terminal and intermediate $\varepsilon$-sequences of $P$.} for $P$, plus $2$.
  If $P=p_1 p_2 \ldots p_n$, a bound for the length of its $\varepsilon$-sequences is given by the sum of the bounds for each prefix $p_i$.
  The bound for the silent, input and output prefixes is 0; the one for the match is 1; a bound for the parallel prefix is the sum of the bounds of its components, plus 1 if all the corresponding $\varepsilon$-sequences are initial or (exclusively) terminal;  a bound for the choice prefix is the maximum of the bounds of its components, plus 1 (usually less since the initial $\varepsilon$-sequences are shrunk here).
\end{proof}

A fundamental characteristic of the proposed semantics is that it yields \emph{ground} transitions, involving only simple labels (no binders, equations, etc.).

\begin{definition} \label{def:lts}
Let $\pi$ be a graph. We denote $\fun{lts}(\pi)\defs\langle Q ,T \rangle$ its labelled transition system 
with $Q$ the set of graphs reachable from $\pi$, and $T$ the set of triplets of the form $(\pi',\alpha,\pi'')$, such that we can infer $\pi'\xrightarrow{\varepsilon^*\alpha}\pi''$, $\alpha\neq \varepsilon$, with the rules of Table~\ref{tab:semantics}. 
\end{definition}

The abstraction from $\varepsilon$-transitions, guaranteed \emph{finitely bound}
by Lemma~\ref{lemma:epsilon:term}, is an important part of the
definition because the normalization steps should not play any
direct behavioral role.

A first - important - step towards finiteness is as follows.

\begin{lemma}
\label{lemma:image:finite}
For any graph $\pi$, $\fun{lts}(\pi)$ is finitely branching. 
\end{lemma}

\begin{proof}
Only the [sum] rule has directly more than one image. By Lemma
\ref{lemma:epsilon:term} the initial $\varepsilon$-sequences for each branch of
the sum have a finite, bounded length, hence there are finitely many of them. 
Moreover, there can be only a finite number of branches in a sum, which bounds the number of images. 

The other source of image-multiplicity is the interleaving of parallel iterators and/or sub-processes, but there are finitely many of them in a $\pi$-graph.
\end{proof}

Based on such (abstracted) labelled
transitions, a ground notion of bisimilarity naturally follows.

\begin{definition} (bisimilarity)\\
Bisimilarity $\bisim$ is the largest symmetric binary relation on $\pi$-graphs such that

$\pi_1\bisim \pi_2$ iff $\pi_1 \xrightarrow{\alpha} \pi '_1 \implies \exists \pi '_2,~\pi_2\xrightarrow{\alpha}\pi '_2$ and $\pi '_1\bisim \pi '_2$ 
\end{definition}

\section{Causal clocks and decidability results}
\label{sec:results}

There are two sources of \emph{infinity} in the basic $\pi$-graph
model. First, the partition $\gamma$ contains initially the singleton
subsets of the infinite set
$\mathcal{N}_f\cup\mathcal{N}_i\cup\mathcal{N}_o$, (only the
free, and fresh input/output names can be made equal). We need a way to
only retain the names that are actually playing a role in the behavior
of the considered $\pi$-graph. Moreover,
logical clocks can evolve infinitely. An example is the fresh name
generator of Figure~\ref{fig:ex:iterator}. To avoid the construction
of infinite state spaces, we first introduce an alternative to logical
clocks.
 
 \begin{definition}\label{def:causal:clock}
A \emph{causal} clock $\kappa$, in the context of an instantiation function $I$, is a partial function in 
$(\{\bot\} \cup \mathcal{N}_o) \rightarrow \powerset{\mathcal{N}_i} $ with
\begin{itemize}
\item $\fun{init} \defs \{\bot \mapsto \emptyset \}$
\item $\fun{out}(\kappa) \defs \kappa \cup \{ \fun{next_o}(\kappa)! \mapsto \emptyset \}$
\item $\fun{in}(\kappa) \defs \{ o \mapsto (\kappa(o) \cup \{\fun{next_i}(\kappa)?\} ) ~|~o\in \fun{dom}(\kappa)\}$
\item $\fun{next_i}(\kappa) \defs \fun{min}\left (\mathbb{N}^+\setminus \{ n\mid n?\in \bigcup \fun{cod}(\kappa)\} \right)$
\item $\fun{next_o}(\kappa) \defs \fun{min}\left (\mathbb{N}^+\setminus \{ n\mid n!\in \fun{dom}(\kappa) \} \right)$
\item $n! \prec_\kappa m? \defs n!\in \fun{dom}(\kappa) \wedge m? \in \kappa(n!)$
\end{itemize}
The names of a clock are $\fun{nm}(\kappa)\defs \fun{dom}(\kappa) \setminus \{\bot\} \cup \bigcup \fun{cod}(\kappa)$.
\end{definition}

Intuitively, $\kappa(n!)$ gathers all the input names $m?$ that were created after $n!$ when the latter was instantiated, 
and $\kappa(\bot)$ gathers all the input names $m?$ that were created, even those that were created before any $n!$. This
 is the minimal amount of information required to record read-write causality on names.

For example, $\fun{next_o}(\fun{init})=1$,  $\kappa\defs\fun{out}(\fun{init})=\{\bot \mapsto \emptyset, 1!\mapsto \emptyset\}$, $\kappa'\defs\fun{in}(\kappa)=\{\bot \mapsto \{1?\}, 1!\mapsto \{1?\}\}$, and $\fun{nm}(\kappa')=\{1!,1?\}$. In $\kappa'$, the input name $1?$ is causally dependent on the output $1!$.

As a second ``counter-measure'' against infinity, we do not record explicitly
(but assume) the singleton sets in the partition. Moreover, we require
the garbage collection for unused names in graphs.

\begin{definition} \label{def:gc} 
The garbage collection $\fun{gc}(\pi)$ of unused names in a graph $\pi$ with causal clock $\kappa$, partition $\gamma$ and instantiations $I$ is $\pi$ with updated clock $\kappa'$ and partition~$\gamma'$ 
such that\\
$\left \{\begin{array}{l}
\gamma' \defs \{ E\cap(\mathcal{N}_f\cup\mathcal{N}_o\cup\fun{cod}(I)) \suchthat E\in \gamma \} \setminus \{\emptyset\} 
\\
\kappa'\defs \{ d \mapsto \kappa(d) \cap \fun{cod}(I) \mid  d\in \fun{dom}(\kappa) \wedge (d=\bot \vee d\in \fun{cod}(I) \vee (\{d\}\not\in\gamma'))\} 
\end{array}\right .$
\end{definition}
For initial graphs, $\fun{gc}(\pi)=\pi$.  The clock only references
instantiated input and output names, plus the output names that are not instantiated but equated to one or more input names. 

From now on we only consider (reachable) garbage-free graphs, i.e.
with unused names implicitly removed. 
This amounts to consider the LTS
$\fun{lts}(\pi)\defs \{ (\pi', \alpha, \fun{gc}(\pi''))\mid 
(\pi',\alpha,\pi'') \mbox{ results from Def.~\ref{def:lts} }\}$. 

\begin{proposition} \label{prop:garbage:free}
Let $\pi$ be a garbage-free graph with clock $\kappa$, partition $\gamma$ and instantiation $I$:
\begin{enumerate}
\item $\fun{dom}(\kappa)= (\fun{cod}(I)\cap \mathcal{N}_o) \cup\{d\in\mathcal{N}_o|\{d\}\not\in\gamma\}\cup\{\bot\}$ and
\item $\bigcup \fun{cod}(\kappa)=\fun{cod}(I)\cap \mathcal{N}_i$ .
\end{enumerate}
Hence $\fun{nm}(\kappa) =(\fun{cod}(I)\cap (\mathcal{N}_o \cup \mathcal{N}_i))\cup\{n!\in \mathcal{N}_o|\{n!\}\not\in\gamma\}$. 
\end{proposition}

\begin{proof}
These are direct consequences of Definition~\ref{def:causal:clock} and Definition~\ref{def:gc}, combined with an induction on the derivation rules.\\
Initially, $\bigcup \fun{cod}(\kappa)=\emptyset=\fun{cod}(I)\cap \mathcal{N}_i$, 
$\fun{dom}(\kappa)= \{\bot\}$, $\fun{cod}(I)\cap \mathcal{N}_o=\emptyset$ and $\gamma$ is only composed of singletons.\\
When a new input name is created by rule [i-fresh], it is added both to $\fun{cod}(I)$ and to $\kappa(\bot)$.\\
When a new output name is created by rule [o-fresh], it is added both to $\fun{cod}(I)$ and to $\fun{dom}(\kappa)$.\\
When an input name is no longer used by $I$, it is suppressed from $\fun{cod}(\kappa)$.\\
When an output name is no longer used by $I$ and it is not equated to some input names, it is suppressed from $\fun{dom}(\kappa)$.
\end{proof}

\begin{proposition}
Causal clocks preserve the freshness constraint.
\end{proposition}

\begin{proof}
  Let $\pi$ be a graph with causal clock $\kappa$ and instantiation
  $I$. By Definition~\ref{def:causal:clock}, we have 
  $\fun{next_o}(\kappa)!\not\in\fun{dom}(\kappa)$
  and $\fun{next_i}(\kappa)?\not\in \bigcup \fun{cod}(\kappa)$. By Proposition~\ref{prop:garbage:free} we
 deduce $\fun{next_o}(\kappa)!\not \in \fun{cod}(I)$ and $\fun{next_i}(\kappa)?\not \in \fun{cod}(I)$ 
\end{proof}

The example of Figure~\ref{fig:ex:iterator} generates, with the logical clocks, an infinite number of states and transitions $\piout{c}{1!},~\piout{c}{2!},\ldots$. 
Using the causal clocks and garbage-free graphs, the behavior collapses to a single state (i.e., a single $\sim$-equivalence class) 
and transition $\piout{c}{1!}$, which is valid because the name $1!$ is not used locally
  and can thus be reused infinitely often. 



We now consider the evolution of the clock along transition
paths from a more general perspective. A fundamental property is that the clock may take only a finite number of values.

\begin{lemma} 
\label{lemma:finite:inputs}
Let a transition system $\fun{lts}(\pi)=\langle Q, T \rangle$ and
consider the causal clock $\kappa_Q$ of each state $Q$: $\bigcup_Q \fun{cod}(\kappa_Q)\subseteq\{1?,2?,\ldots,|B|?\}$.
\end{lemma}

\begin{proof}
  First, a direct consequence of Proposition~\ref{prop:garbage:free}(2) is that $|\bigcup_Q \fun{cod}(\kappa_Q)|\leq |B|$, since $|\fun{cod}(I)|\leq|\fun{dom}(I)|=|B|$.
  Initially, $\bigcup_Q \fun{cod}(\kappa_Q)=\emptyset$.
  The unique way to increase the size of the codomain of a clock (by one) is
  through an [i-fresh] transition. If, at that point, $k$ is the first integer such that $k?$ is missing in $\bigcup_Q \fun{cod}(\kappa_Q)$, it will be added to it. Thus we shall have either $\bigcup_Q \fun{cod}(\kappa_Q)=\{1?,2?,\ldots,(k-1)?,(k+h)?\ldots\}$ becomes $\{1?,2?,\ldots,(k-1)?,k?,(k+h)?\ldots\}$ or $\bigcup_Q \fun{cod}(\kappa_Q)=\{1?,2?,\ldots,(k-1)?\}$ becomes $\{1?,2?,\ldots,(k-1)?,k?\}$. Hence the property.
\end{proof}

For the fresh outputs the situations is similar, but for a slightly different reason.

\begin{lemma} 
\label{lemma:finite:outputs}
Let a transition system $\fun{lts}(\pi)=\langle Q, T \rangle$ and
consider the causal clock $\kappa_Q$ of each state $Q$: $\fun{dom}(\kappa_Q)\subseteq\{\bot,1!,2!,\ldots,|B|!\}$.
\end{lemma}

\begin{proof}
  From Proposition~\ref{prop:garbage:free}, we know that $\fun{dom}(\kappa_Q)$ always contains $\bot$ and the instantiated output names; let us assume there are $k$ of the latter; there are thus at most $|B|-k$ instantiated input names; now each non-instantiated output name may only be equated by $\gamma$ to instantiated input names and there is no intersection between the latter; hence there are at most $|B|-k$ non-instantiated output names left in $\fun{dom}(\kappa_Q)$. Then, the reasoning is similar to the one for Lemma~\ref{lemma:finite:inputs}.
\end{proof}



\begin{lemma} \label{lemma:finite:states} 
Let $\pi$ be a graph with causal clocks, and $\fun{lts}(\pi)=\langle Q,T \rangle$ its corresponding transition system.
The sets $Q$ and $T$ are of finite size.
\end{lemma}

\begin{proof}
  Each state of $Q$ is a reachable configuration following
  Definition~\ref{def:pigraph}. Infinity can only result from the
  parts of the configuration that evolve along transitions, i.e., the
  clock $\kappa$, the partition $\gamma$, the instantiation $I$ and
  the marking
  $M$. There is a finite bound
  for the number of possible markings ($2^p$ where $p$ is the number
  of places in the configuration). Lemmas
  ~\ref{lemma:finite:inputs} and~\ref{lemma:finite:outputs} assert that the set of reachable
  (causal) clocks is also finite. For the instantiation $I$, only the number of input and    
  output fresh names may increase. We can deduce from Proposition~\ref{prop:garbage:free} 
  that $\fun{cod}(I)\cap (\mathcal{N}_i\cup\mathcal{N}_o)\subseteq \fun{nm}(\kappa)$
  and thus the set of reachable instantiations is also finite. We can then observe that, from the previous definitions, the non-singleton classes in a partition only contain names in $\fun{cod}(I)\cup\fun{nm}(\kappa)$,
hence the number of reachable
 partitions is finite. In consequence there are only finitely many reachable configurations, thus $Q$ is finite.
 Finally, by Lemma~\ref{lemma:image:finite} we know that $T$ is image-finite and a finitely branching relation
over the finite set $Q$ is finite.
\end{proof}

\begin{theorem}
Bisimilarity for $\pi$-graphs with causal clocks is decidable
\end{theorem}

This important result is a direct consequence of Lemma~\ref{lemma:finite:states}.

\section{Related work}
\label{sec:related}


The design of visual languages for mobile systems has been
investigated in Milner's $\pi$-nets~\cite{DBLP:conf/esop/Milner94} and
Parrow's interaction diagrams~\cite{DBLP:journals/njc/Parrow95}.  The
$\pi$-graphs try to convey the ``inventiveness'' of such attempts but
building on more formal grounds and with an emphasis on
\emph{practicability} from a modelling perspective.  The main
characteristic of our formalism, from this point of view, is the fact
that the structure of the graphs remains \emph{static} along 
transitions. This is a major difference when compared to other graphic
variants of the $\pi$-calculus~\cite{gadducci:pigraphenc:mscs:2006},
including the \emph{dynamic}
$\pi$-graphs~\cite{DBLP:conf/sofsem/PeschanskiB09}. From a technical
standpoint this design choice has a profound impact on the
semantics. Instead of relying on more expressive graph rewriting
techniques~\cite{gadducci:pigraphenc:mscs:2006,DBLP:journals/iandc/BonchiGK09},
we exploit an inductive variant of \emph{graph
  relabelling}~\cite{graph:relabelling:handbook}.  The inductive
extension is used to characterize the choice
operator.  A lower-level implementation is possible (see
e.g.~\cite{devillers-klaudel-koutny08}) but inductive rules provide a much more
 concise characterization.


Similarly to Petri nets, the motivation behind the $\pi$-graphs is not
limited to \emph{modelling} purposes. The formalism should be
\emph{also} suitable for the automated verification of mobile systems.
There are indeed only a few verification techniques and tools
developed for the $\pi$-calculus and variants. Decision procedures for
open bisimilarity are proposed in
e.g.~\cite{DBLP:journals/iandc/PistoreS01,DBLP:conf/cav/VictorM94}. The
techniques developed are not trivial and specific to the
$\pi$-calculus (or also the fusion calculus in recent versions
of~\cite{DBLP:conf/cav/VictorM94}). In comparison, the $\pi$-graphs
rely on ground notions of transition and bisimulation, which means
standard techniques and existing tools can be directly employed. There
is a connection between the symbolic semantics used to characterize
open bisimilarity and the partition $\gamma$ in the $\pi$-graph
configurations. Instead of recording equalities in transitions, we
record the effect of the equality directly in the states. This means
it is never required to ``go back in time'' to recover a particular
equality. Moreover, we think a similar mechanism can be used to
implement the mismatch construct. Open bisimilarity enjoys a much
desired congruence property. It remains an open question whether
bisimilarity on $\pi$-graphs is a congruence or not. We conjecture
this is the case, e.g.  $\piin{a}{x}[x=b]\piout{b}{c}$ and
$\piin{a}{x}0$ are properly discriminated. However the formal proof is
left as a future work.

Another approach is to translate some $\pi$-calculus variant into
another formalism with better potential for verification. A positive
aspect is that this makes the verification framework (relatively)
independent from the source language. The other side of the coin is
that it is more difficult to connect the verification results
(e.g. counter-examples) to the modelling formalism.  The early
labelled transition systems for the $\pi$-calculus can be translated
to history dependent automata
(HDA)~\cite{DBLP:conf/concur/MontanariP95,Montanari-Pistore-HD}.  The
states of HDA contain the sets of active (restricted) names, and the
transitions provide injective correspondences so that names can be
created and, most importantly, forgotten. This gives a local
interpretation of freshness whereas the $\pi$-graphs use a global
interpretation using clocks. Unlike HDA, the problem of garbage
collecting unused names in $\pi$-graphs can be decided by inspecting
the current state of the computation. HDA is an intermediate
semantic-level formalism. They are produced from process expressions
and can in turn be unfolded as plain automata. With $\pi$-graph, we
are able to produce basic (ground) automata directly.

There are also various translations of Pi-calculus variants
into Petri nets.  In~\cite{devillers-klaudel-koutny08} we propose a
translation of the $\pi$-calculus into finite high-level Petri nets
(with read arcs), using basic net composition operators.  Beyond the
use of a high-level (and Turing-complete) model of Petri nets, another
issue we face is the encoding of recursive behaviors as unfolding.
Indeed, the verification problems are only decidable for
recursion-free processes in this
framework. In~\cite{DBLP:conf/concur/MeyerG09} an alternative
translation to lower-level P/T nets is proposed. The translated nets
cannot be used as modelling artifacts. First, they may have a size
exponentially larger than the initial $\pi$-calculus terms. Moreover
their structure does not reflect the structure of the terms but
corresponds to behavioral properties: the places are connection
patterns and the tokens instances of these patterns.  However, the
translation is particularly suitable for the verification
problem. Indeed, the translated P/T nets have a finite size for a
class of \emph{structural stationary} systems, which is strictly
larger than finite-control processes. Note, however, that the
membership problem for this class is undecidable. Moreover, it is not
a compositional property. The iterator construct is slightly more
expressive than the finite-control class of processes. The latter can
be encoded using iterators and the communication primitives. But it is
also possible to encode behaviors in which the number of active
threads changes along iterations (although their number must be
bound). Unlike the $\pi$-graphs, only the reduction semantics for
closed systems are considered in~\cite{DBLP:conf/concur/MeyerG09}. As
explained in~\cite{gadducci:pigraphenc:mscs:2006}, the switch from the
reduction to the transition semantics is not trivial.  Recent works,
e.g.~\cite{DBLP:journals/iandc/BonchiGK09}, suggest the use of
\emph{borrowed contexts} (BC) to derive transition systems (and
bisimulation congruence) from graph grammars. In the $\pi$-graphs, we
propose an alternative technique of deriving transition labels from
node attributes, which we find simpler. However, we cannot
derive any congruence result from the construction. To our knowledge
the $\pi$-calculus has not been fully characterized in the BC
framework.

\section{Conclusion and future work}

The $\pi$-graphs is a visual paradigm for the modelling and
verification of mobile systems. It has constructs very close to the
$\pi$-calculus, although strictly speaking it is more a variant than a
graphical encoding. We plan to establish stronger connections between
(traditional) variants of the $\pi$-calculus and the $\pi$-graphs. In
particular, we conjecture $\pi$-graph bisimilarity to be close to late
congruence. For the latter, it seems cumbersome to work directly with the
$\pi$-graphs, because they involve relatively complex process
contexts. A privileged direction would be to translate the graphs back
into a variant of the $\pi$-calculus, and study the meta-theory at
that level.

For verification purposes, the $\pi$-graphs with iterators enjoy appealing
properties: the semantics rely on ground notions of transition and
bisimulation, and their state-space is finite by construction. However
the size of the LTS can be exponentially larger than the initial
graphs. To cope with this state explosion problem, we plan to
complement the traditional interleaving semantics developed in this
paper by more causal semantics. An interesting approach is to
\emph{slice} the semantics by analyzing independently each iterator of
a graph. Instead of interleaving the slices it is possible to relate
them in a causal way, considering the fact that the only transitions
across iterators are synchronizations.  Seen as an intermediate model,
the $\pi$-graphs - in particular the iterator construct - offer a
major simplification to our own Petri net translation of the
$\pi$-calculus~\cite{devillers-klaudel-koutny08}.  We think a
lower-level Petri net model can be used in the translation, with better
dispositions for verification using existing Petri net tools.


Last but not least, we plan to integrate the static variant of the 
$\pi$-graphs, as presented in this paper, in our prototype
tool available online\footnote{cf. \url{http://lip6.fr/Frederic.Peschanski/pigraphs}.}.

\bibliographystyle{eptcs}
\bibliography{pigraphs}

\end{document}